\documentclass[a4paper,UKenglish,cleveref,autoref]{lipics-v2019}

\usepackage{amsfonts}
\usepackage{amsmath}
\usepackage{amssymb}
\usepackage{graphicx}%
\usepackage[normalem]{ulem}

\usepackage{todonotes}
\usepackage{algorithm}

\usepackage{xcolor}

\newcommand{\cnpg}[1]{cnp({#1})}
\newcommand{\distgg}{d_{GG}}
\newcommand{\distgcnp}{d_{GCNP}}
\newcommand{\mcng}{\textsf{MCNG}}

\renewcommand{\c}{\vec{c}}
\renewcommand{\u}{\vec{u}}
\renewcommand{\v}{\vec{v}}
\newcommand{\w}{\vec{w}}
\newcommand{\s}[1]{\langle \beta_{#1} \rangle}
\newcommand{\e}[1]{\alpha_{#1}}
\newcommand{\gevent}[1]{\langle {#1} \rangle}

\renewcommand{\S}{{\mathcal{S}}}

\newcommand{\cnp}[1]{\langle {#1} \rangle}

\newtheorem{nclaim}{Claim}

\title{Genomic Problems Involving Copy Number Profiles: Complexity and Algorithms}

\titlerunning{Genomic Problems Involving Copy Number Profiles}
\author{Manuel Lafond}{Department of Computer Science, Universite de Sherbrooke, Sherbrooke, Quebec J1K 2R1, Canada}{manuel.lafond@usherbrooke.ca}{}{}
\author{Binhai Zhu}{Gianforte School of Computing, Montana State University, Bozeman, MT 59717, USA}{bhz@montana.edu}{}{}
\author{Peng Zou}{Gianforte School of Computing, Montana State University, Bozeman, MT 59717, USA}{peng.zou@student.montana.edu}{}{}

\authorrunning{Lafond, Zhu and Zou.}
\Copyright{Manuel Lafond, Binhai Zhu and Peng Zou}

\ccsdesc{Theory of computation}

\keywords{Computational genomics, cancer genomics, copy number profiles, NP-hardness, approximation algorithms, FPT algorithms}

\category{}

\relatedversion{}

\supplement{}

\funding{This research was partially supported by NSERC of Canada.}
\EventEditors{Inge Li G{\o}rtz and Oren Weimann}
\EventNoEds{2}
\EventLongTitle{31st Annual Symposium on Combinatorial Pattern Matching (CPM 2020)}
\EventShortTitle{CPM 2020}
\EventAcronym{CPM}
\EventYear{2020}
\EventDate{June 17-19, 2020}
\EventLocation{Copenhagen, Denmark}
\EventLogo{}
\SeriesVolume{159}
\ArticleNo{1} % ¡°New number¡± (=<article-no>) goes here!
\nolinenumbers %uncomment to disable line numbering
\hideLIPIcs  %uncomment to remove references to LIPIcs series (logo, DOI, ...), e.g. when preparing a pre-final version to be uploaded to arXiv or another public repository
\begin{document}
\maketitle  
\begin{abstract}
Recently, due to the genomic sequence analysis in several types of cancer,
the genomic data based on {\em copy number profiles} ({\em CNP} for short)
are getting more and more popular. A CNP is a vector where each
component is a non-negative integer representing the number of copies of a specific
gene or segment of interest. The motivation is that in the late stage of certain types of cancer,
the genomes are progressing rapidly by segmental duplications and deletions
hence obtaining the exact sequences is becoming more difficult. Instead, in
this case, the number of copies of important genes can be predicted from expression analysis and carries important biological
information. Therefore, a lot of research have been started to analyze the
genomic data represented in CNP's.

In this paper, we present two streams of results. The first is the negative
results on two open problems regarding the computational complexity of the
Minimum Copy Number Generation (MCNG) problem posed by Qingge et al. in
2018. The \emph{Minimum Copy Number Generation} (MCNG) is defined as follows:
given a string $S$ over a gene set $\Sigma$ (with $|\Sigma|=n$) and a 
CNP $C$, compute a string $T$ from $S$, with the minimum number of
segmental duplications and deletions, such that $cnp(T)=C$. It was shown by
Qingge et al. that the problem is NP-hard if the duplications are tandem and
they left the open question of whether the problem remains NP-hard if
arbitrary duplications are used. We answer this question affirmatively in this
paper; in fact, we prove that it is NP-hard to even obtain a constant factor
approximation.   This is achieved through a general-purpose lemma on set-cover reductions that require an exact cover in one direction, but not the other, which might be of independent interest.  We also prove that the corresponding parameterized version is
W[1]-hard, answering another open question by Qingge et al.

The other result is positive and is based on a new (and more general)
problem regarding CNP's. The \emph{Copy Number Profile Conforming (CNPC)}
problem is formally defined as follows: given two CNP's $C_1$ and $C_2$,
compute two strings $S_1$ and $S_2$ with $cnp(S_1)=C_1$ and $cnp(S_2)=C_2$
such that the distance between $S_1$ and $S_2$, $d(S_1,S_2)$, is minimized.
Here, $d(S_1,S_2)$ is a very general term, which means it could be any genome
rearrangement distance (like reversal, transposition, and tandem
duplication, etc). We make the first step by showing that if $d(S_1,S_2)$ is
measured by the breakpoint distance then the problem is polynomially solvable.
We expect that this will trigger some related research along the line in the
near future.
\end{abstract}

\section{Introduction}

In cancer genomics research, intra-tumor genetic heterogeneity is one of the
central problems \cite{marusyk2012intra,navin2010inferring,shah2009mutational}.
Understanding the origins of cancer cell diversity could help cancer prognostics \cite{cooke2011evolution,maley2006genetic} and
also help explain drug resistance \cite{cooke2011intra,cowin2012lrp1b}. It is
known for some types of cancers, such as
high-grade serous ovarian cancer (HGSOC), that heterogeneity is mainly acquired through
endoreduplications and genome rearrangements. These result in aberrant
{\em copy number profiles (CNPs)} --- nonnegative integer vectors representing
the numbers of genes occurring in a genome \cite{cancer2011integrated}.
To understand how the cancer progresses, an evolutionary tree is certainly
desirable, and producing a valid evolutionary tree based on these
genomic data becomes a new problem. In \cite{schwarz2014phylogenetic}, Schwarz
et~al. proposed a way to construct a phylogenetic tree directly from integer
copy number profiles, the underlying problem being to convert CNPs into one another using the
minimum number of duplications/deletions \cite{shamir2016linear}.

In \cite{qingge2018}, another fundamental problem was proposed. The motivation
is that in the early stages of cancer, when large numbers of endoreduplications
are still rare, genome sequencing is still possible. However, in the later
stage we might only be able to obtain genomic data in the form of CNPs. 
This leads to the problem of comparing a sequenced genome with a genome with only copy-number information.
%In
%this case, we need to investigate a slightly different algorthmic problem
%called MCNG, which is defined as follows. 

Given a genome $G$ represented as a string and a copy number profile $\c$, the
{\em Minimum Copy Number Generation (MCNG)} problem asks for the minimum number of
deletions and duplications needed to transform $G$ into any genome in which
each character occurs as many times as specified by $\c$.   
Qingge et al. proved that the problem is NP-hard when the duplications are
restricted to be tandem and posed several open questions: (1) Is the problem
NP-hard when the duplications are arbitrary? (2) Does the problem admit
a decent approximation? (3) Is the problem fixed-parameter tractable (FPT)? 
In this paper, we answer all these three open questions. We show that MCNG is
NP-hard to approximate within any constant factor, and that it is W[1]-hard
when parameterized by the solution size.
The inapproximability follows from a new general-purpose lemma on set-cover reductions that require an exact cover in one direction, but not the other. The W[1]-hardness uses a new set-cover variant in which every optimal solution is an exact cover.  These set-cover extensions can make reductions to other problems easier, and may be of independent interest.

We also consider a new fundamental problem called {\em Copy Number Profile
Conforming (CNPC)}, which is defined as follows. Given two CNP's $\c_1$ and
$\vec{c}_2$, compute two strings/genomes $S_1$ and $S_2$ with $cnp(S_1)=\c_1$ and
$cnp(S_2)=\c_2$ such that the distance between $S_1$ and $S_2$, $d(S_1,S_2)$,
is minimized. The distance $d(S_1,S_2)$ could be general, which means it could
be any genome rearrangement distance (such as reversal, transposition, and
tandem duplication, etc). We make the first step by showing that if
$d(S_1,S_2)$ is measured by the breakpoint distance then the problem is
polynomially solvable.

%This \todo{Consider removing this paragraph is we need room.} paper is organized as follows. In Section 2, we give basic definitions
%and the corresponding problems. In Section 3, we present the inapproximability
%hardness results for MCNG. In Section 4, we prove the W[1]-hardness for
%MCNG. In Section 5, we present a polynomial time algorithm for CNPC. We
%conclude the paper with some further open problems in Section 6.

\section{Preliminaries}

A genome $G$ is a string, i.e. a sequence of characters, all of which belong
to some alphabet $\Sigma$ (the characters of $G$ can be interpreted as genes
or segments). We use genome and string interchangeably in this paper, when
the context is clear.  A \emph{substring} of $G$ is a sequence of contiguous characters
that occur in $G$, and a \emph{subsequence} is a string that can be obtained
from $G$ by deleting some characters.  We write $G[p]$ to denote the character
at position $p$ of $G$ (the first position being $1$), and we write $G[i .. j]$
for the substring of $G$ from positions $i$ to $j$, inclusively.  For
$s \in \Sigma$, we write $G - s$ to denote the subsequence of $G$ obtained by removing all occurrences of $s$.

We represent an alphabet as an ordered list $\Sigma = (s_1, s_2, \ldots, s_m)$ of distinct characters.  Slightly abusing notation, we may write $s \in \Sigma$ if $s$ is a member of this list.
We write $n_s(G)$ to denote the number of occurrences of $s \in \Sigma$ in a genome $G$.
A \emph{Copy-Number Profile (or CNP)} on $\Sigma$ is a vector $\c = \cnp{c_1, \ldots, c_{|\Sigma|}}$ that associates each character $s_i$ of the alphabet with a non-negative integer $c_i \in \mathbb{N}$.  
We may write $\c(s)$ to denote the number associated with $s \in \Sigma$ in $\c$.
We write $\c - s$ to denote the CNP obtained from $\c$ by setting $\c(s) = 0$.

The \emph{Copy Number Profile (CNP)} of genome $G$, denoted $\cnpg{G}$, is the vector of occurrences of all characters of $\Sigma$.   
Formally\footnote{Note that in the theory of formal languages, the CNP of a string is called the \emph{Parikh vector}.}, 
$$\cnpg{G} = \cnp{ n_{s_1}(G), n_{s_2}(G), \ldots, n_{s_m}(G) }.$$
For example, if $\Sigma = (a,b,c)$ and $G = abbcbbcca$, 
then $\cnpg{G} = \cnp{2, 4, 3}$ and $\c(a)=2$.
\newline

\noindent
\textbf{Deletions and duplications on strings}

\noindent
We now describe the two string events of \emph{deletion} and \emph{duplication}.  Both are illustrated in Figure~\ref{fig:example-seqs}.

\vspace{-5mm}

\begin{figure}[htbp]
\centering
\begin{eqnarray*}
Sequence & ~~~~~~~~~~~~~Operations
\\
 G_1= abbc \cdot \text{\sout{\ensuremath{cab}}}\cdot cab & ~~~~~~~~del(4,6)
\\
 G_2=a\cdot \text{\underline{\ensuremath{bbcc}}}\cdot ab & ~~~~~~~~~~~dup(2,5,6)
\\
G_3=abbcca\cdot \text{\underline{\ensuremath{bbcc}}}\cdot b & ~~
\end{eqnarray*}
\caption{Three strings (or toy genomes), $G_1, G_2$ and $G_3$.
From $G_1$ to $G_2$, a deletion is applied to $G_1[4..6]$.  From $G_2$ to $G_3$, a duplication is applied to $G_2[2..5]$, with the copy  inserted after position 6.}
\label{fig:example-seqs}
\end{figure}

Given a genome $G$, a \emph{deletion} on $G$ takes a substring of $G$ and removes it.  Deletions are denoted by a pair $(i, j)$ of the positions of the substring to remove.  
Applying deletion $(i, j)$ to $G$ transforms 
$G$ into $G[1 .. i - 1]G[j + 1 .. n]$.

A \emph{duplication} on $G$ takes a substring of $G$, copies it and inserts the copy anywhere in $G$, except inside the copied substring.  A duplication is defined by a triple $(i, j, p)$ where $G[i .. j]$ is the string to duplicate and $p \in \{0, 1, \ldots, i - 1, j, \ldots n\}$ is the position \emph{after} which we insert (inserting after $0$ prepends the copied substring to $G$).
Applying duplication $(i, j, p)$ to $G$ transforms $G$ into 
$G[1 .. p]G[i .. j]G[p + 1 .. n]$.
%, with the constraint that $p < i$ or $p \geq j$.
%(if $p = 0$, $G[1 .. p]$ is defined as the empty string).  

An \emph{event} is either a deletion or a duplication. If $G$ is a genome and $e$ is an event , we write $G\gevent{e}$ to denote the genome obtained by applying $e$ on $G$.  Given a sequence $E = (e_1, \ldots, e_k)$ of events, we define $G\gevent{E} = G\gevent{e_1}\gevent{e_2} \ldots \gevent{e_k}$ as the genome obtained by successively applying the events of $E$ to $G$.  We may also write $G\gevent{e_1, \ldots, e_k}$ instead of $G\gevent{(e_1, \ldots, e_k)}$.

The most natural application of the above events is to compare genomes.

\begin{definition}
Let $G$ and $G'$ be two strings over alphabet $\Sigma$.  The \emph{Genome-to-Genome} distance between $G$ and $G'$, denoted $\distgg(G, G')$, is the size of the smallest sequence of events $E$ satisfying $G\gevent{E} = G'$.
\end{definition}

We also define a distance between a genome $G$ and a CNP $\c$, which is the minimum number of events to apply to $G$ to obtain a genome with CNP $\c$.

\begin{definition}
Let $G$ be a genome and $\c$ be a CNP, both over alphabet $\Sigma$.  The \emph{Genome-to-CNP} distance between $G$ and $\c$, denoted $\distgcnp(G, \c)$, is the size of the smallest sequence of events $E$ satisfying $\cnpg{G\gevent{E}} = \c$. 
\end{definition}

The above definition leads to the following problem, which was first studied in \cite{qingge2018}. 

\vspace{3mm}

\noindent
The \textsf{Minimum Copy Number Generation  (MCNG)} problem:
\\
\noindent
\emph{Instance:} a genome $G$ and a CNP $\c$ over alphabet $\Sigma$.\\
\noindent
\emph{Task:} compute $\distgcnp(G, \c)$.

Qingge et al. proved that the MCNG problem is NP-hard when all the duplications
are restricted to be tandem \cite{qingge2018}. In the next section, we prove
that this problem is not only NP-hard, but also NP-hard to approximate
within any constant factor.

\section{Hardness of Approximation for MCNG}

In this section, we show that the $\distgcnp$~distance is hard to approximate within any constant factor.  
This result actually holds if only deletions on $G$ are allowed.  This restriction makes the proof significantly simpler, so we first analyze the deletions-only case.  We then extend this result to deletions \emph{and} duplications.

Both proofs rely on a reduction from SET-COVER.  
Recall that in SET-COVER, we are given a collection of sets $\S = \{S_1, S_2, \ldots, S_n\}$ over universe $U = \{u_1, u_2, \ldots, u_m\} = \bigcup_{S_i \in \S}S_i$, and we are asked to find a set cover of $\S$ having minimum cardinality (a set cover of $\S$ is a subset $\S^* \subseteq \S$ such that $\bigcup_{S \in \S^*}S = U$).  If $\S'$ is a set cover in which no two sets intersect, then $\S'$ is called an \emph{exact cover}.

There is one interesting feature (or constraint) of our reduction $g$, which transforms a SET-COVER instance $\S$ into a MCNG instance $g(\S)$.  A set cover $\S^*$ only works on $g(\S)$ if $\S^*$ is actually an exact cover, and a solution for $g(\S)$ can be turned into a set cover for $\S^*$ that is not necessarily exact.  Thus we are unable to reduce directly from either SET-COVER nor its exact version.  We provide a general-purpose lemma for such situations, and our reductions serve as an example of its usefulness.

The proof relies on a result on $t$-SET-COVER,the special case of SET-COVER in which every given set contains at most $t$ elements.  It is known that for any constant $t \geq 3$, the $t$-SET-COVER problem is hard to approximate within a factor $\ln t - c \ln \ln t$ for some constant $c$ not depending on $t$
\cite{trevisan2001}.

\begin{lemma}\label{lem:general-inapprox}
Let $\mathcal{B}$ be a minimization problem, and 
let $g$ be a function that transforms any SET-COVER instance $\S$ into an instance $g(\S)$ of $\mathcal{B}$ in polynomial time. 
Assume that both the following statements hold:
\begin{itemize}
    \item 
    any exact cover $\S^*$ of $\S$ of cardinality at most $k$ can be transformed in polynomial time into a solution of value at most $k$ for $g(\S)$;
    
    \item 
    any solution of value at most $k$ for $g(\S)$ can be transformed in polynomial time into a set cover of $\S$ of cardinality at most $k$.
    
\end{itemize}

Then unless P = NP, there is no constant factor approximation algorithm for $\mathcal{B}$.
\end{lemma}

\begin{proof}
Suppose for contradiction that $\mathcal{B}$ admits a factor $b$ approximation for some constant $b$.  Choose any constant $t$ such that $t$-SET-COVER is hard to approximate within factor $\ln t - c \ln \ln t$, and such that
$b < \ln t - c \ln \ln t$.  Note that $t$ might be exponentially larger than $b$, but is still a constant.

Now, let $\S$ be an instance of $t$-SET-COVER over the universe $U = \{u_1, \ldots, u_m\}$.  
Consider the intermediate reduction $g'$ that transforms $\S$ into another $t$-SET-COVER instance $g'(\S) = \{ S' \subseteq S : S \in \S \}$.  
Since $t$ is a constant, $g(\S)$ has $O(|\S|)$ sets and this can be carried out in polynomial time.  

Now define $\S' = g'(\S)$ and consider the instance $B = g(\S') = g(g'(\S))$.
By the assumptions of the lemma, a solution for $B$ of value $k$ yields a set cover $\S^*$ for $\S'$.  Clearly, $\S^*$ can be transformed into a set cover for instance $\S$: for each $S' \in \S^*$, there exists  $S \in \S$ such that $S' \subseteq S$, so we get a set cover for $\S$ by adding this corresponding superset for each $S \in \S^*$.  Thus $B$ yields a set cover of $\S$ with at most $k$ sets.

In the other direction, consider a set cover $\S^* = \{S_1, \ldots, S_k\}$ of $\S$ with $k$ sets.  This easily translates into an \emph{exact} cover of $\S'$ with $k$ sets by taking the collection 

$$\{S_1, S_2 \setminus S_1, S_3 \setminus (S_1 \cup S_2), \ldots, S_k \setminus \bigcup_{i = 1}^{k-1} S_i\}\}.
$$

By the assumptions of the lemma, this exact cover can then be transformed into a solution of value at most $k$ for instance $B$.

Therefore, $\S$ has a set cover of cardinality at most $k$ if and only if $B$ has a solution of value at most $k$.  Since there is a correspondence between the solution values of the two problems, a factor $b$ approximation for $\mathcal{B}$ would provide a factor $b < \ln t - c \ln \ln t$ approximation for $t$-SET-COVER.  
%\qed
\end{proof}

%Towards a reduction from SET-COVER, we then describe how to construct a genome from a collection of sets, and use this construction to obtain our inapproximability results.  A variant of the same construction will then be used to prove the W[1]-hardness of \mcng.

\subsection{Constructing genomes and CNPs from SET-COVER instances}\label{subsec:construction}

All of our hardness results rely on Lemma~\ref{lem:general-inapprox}.  We need to provide a reduction from SET-COVER to~\mcng~and prove that both assumptions of the lemma are satisfied.

\begin{figure}
\centering
\[
    S_1 = \{1,2,3\} \quad S_2 = \{1,3,4\} \quad S_3 = \{2,3,5\}
\]
\[
G = \s{S_1}\e{1}\e{2}\e{3}\s{S_2}\e{1}\e{3}\e{4}\s{S_3}\e{2}\e{3}\e{5}
\]
\[
\c(\e{1}) = \c(\e{2}) = 1 \quad \c(\e{3}) = 2 \quad \c(\e{4}) = \c(\e{5}) = 0
\]
\caption{An example of our construction, with $\S = \{S_1, S_2, S_3\}$ and $U = \{1,2,3,4,5\}$.}\label{fig:reduction_example}

\end{figure}

This reduction is the same for deletions-only and deletions-and-duplications.  Given $\S$ and $U$, we construct a genome $G$ and a CNP $\c$ as follows (an example is illustrated in Figure~\ref{fig:reduction_example}).
The alphabet is $\Sigma = \Sigma_{\S} \cup \Sigma_U$, where ${\Sigma_{\S} := \{\s{S_i} : S_i \in \S \}}$ and ${\Sigma_U := \{\e{u_i} : u_i \in U\}}$.  Thus, there is one character for each set of $\S$ and each element of $U$.
Here, each $\s{S_i}$ is a character that will serve as a separator between characters to delete. 
For a set $S_i \in \S$, define the string $q(S_i)$ as any string that contains each character of  
$\{\e{u} : u \in S_i\}$ exactly once.  We put 

$$
G = \s{S_1}q(S_1)\s{S_2}q(S_2) \ldots \s{S_n}q(S_n),
$$
i.e. $G$ is the concatenation of the strings $\s{S_i}q(S_i)$.
As for the CNP $\c$, put 

\begin{itemize}
    \item 

$\c(\s{S_i}) = 1$ for each $S_i \in \S$;

    \item 
$\c(\e{u}) = f(u) -1$ for each $u \in U$, where $f(u) = |\{S_i \in \S : u \in S_i\}|$ is the number of sets from $\S$ that contain $u$.
\end{itemize}

Notice that in $G$, each $\s{S}$ already has the correct copy-number, whereas each $\e{u}$ needs exactly one less copy.  Our goal is thus to reduce the number of each $\e{u}$ by $1$.
This concludes the construction of~\mcng~instances from SET-COVER instances.  We know focus on the hardness of the deletions-only case.

\subsection{Warmup: the deletions-only case}

Suppose that we are given a set cover instance $\S$ and $U$, and let $G$ and $\c$ be the genome and CNP, respectively, as constructed above.

\begin{lemma}\label{lem:exact-to-dels}
Given an exact cover $\S^*$ for $\S$ of cardinality $k$, one can obtain a sequence of $k$ deletions transforming $G$ into a genome with CNP $\c$.
\end{lemma}

\begin{proof}
Denote $\S^* = \{S_{i_1}, \ldots, S_{i_k}\}$.
Consider the sequence of $k$ deletions that deletes the substrings $q(S_{i_1}), \ldots, q(S_{i_k})$ (i.e. the sequence first deletes the substring 
$q(S_{i_1})$, then deletes $q(S_{i_2})$, and so on until $q(S_{i_k})$ is deleted).  Since $S_{i_1}, \ldots, S_{i_k}$ is an exact cover, 
this sequence removes exactly one copy of each $\e{u} \in \Sigma_U$ and does not affect the $\s{S}$ characters.
Thus the $k$ deletions transform $G$ into a genome with the desired CNP $\c$.
%\qed
\end{proof}

\begin{lemma}\label{lem:events-to-cover}
Given a sequence of $k$ deletions transforming $G$ into a genome with CNP $\c$, one can obtain a set cover for $\S$ of cardinality at most $k$.
\end{lemma}

\begin{proof}
Suppose that the deletion events $E = (e_1, \ldots, e_k)$  transform $G$ into a genome 
$G^*$ with CNP $\c$.  
Note that no $e_i$ deletion is allowed to delete a set-character $\s{S_i} \in \Sigma_{\S}$, as there is only one occurrence of $\s{S_i}$ in $G$ and $\c(\s{S_i}) = 1$.
Thus all deletions remove only $\e{u}$ characters.  In other words, each $e_j$ in $E$ either deletes a substring of $G$ between some $\s{S_i}$ and $\s{S_{i+1}}$ with $1 \leq i < n$, or $e_j$ deletes a substring after $\s{S_n}$.  Moreover, exactly one of each $\e{u}$ occurrences gets deleted from $G$.

Call $\s{S_i} \in \Sigma_{\S}$ \emph{affected} if there is some event of $E$ that deletes at least one character between $\s{S_i}$ and $\s{S_{i+1}}$ with $1 \leq i < n$, and call $\s{S_n}$ affected if some event of $E$ deletes characters after $\s{S_n}$.
Let {$\S^* := \{S_i \in \S : \s{S_i}$~is~affected$\}$}.
Then $|\S^*| \leq k$, since each deletion affects at most one $\s{S_i}$ and there are $k$ deletion events.
Moreover, $\S^*$ must be a set cover, because each $\e{u} \in \Sigma_U$ has at least one occurrence that gets deleted and thus at least one set containing $u$ that is included in $\S^*$.  This concludes the proof.
%\qed
\end{proof}

We have shown that all the assumptions required by Lemma~\ref{lem:general-inapprox} are satisfied.  The inapproximability follows.

\begin{theorem}\label{thm:hard-dels-only}
Assuming $P \neq NP$, there is no polynomial-time constant factor approximation algorithm for~\mcng~when only deletions are allowed.
\end{theorem}

We mention without proof that the reduction should be adaptable to the duplication-only case, by putting $\c(\e{u}) = f(u) + 1$ for each $u \in U$.

\subsubsection*{The real deal: deletions and duplications}

We now consider both deletions and duplications.  The reduction uses the same construction as in Section~\ref{subsec:construction}.  Thus we assume that we have a SET-COVER instance $\S$ over $U$, and a corresponding instance of~\mcng~with genome $G$ and CNP $\c$. 

In that case, we observe the following: Lemma~\ref{lem:exact-to-dels} still holds whether we allow deletion only, or both deletions and duplications.  Thus we only need to show that the second assumption of Lemma~\ref{lem:general-inapprox} holds.

Unfortunately, this is not as simple as in the deletions-only case.  The problem is that some duplications may copy some $\e{u}$ and $\s{S_i}$ occurrences, and we lose control over what gets deleted, and over what $\s{S_i}$ each $\e{u}$ corresponds to (in particular, some $\s{S_i}$ might now get deleted, which did not occur in the deletions-only case).

Nevertheless, the analogous result can be shown to hold.

\begin{lemma}\label{lem:events-to-cover-dup}
Given a sequence of $k$ events (deletions and duplications) transforming $G$ into a genome with CNP $\c$, one can obtain a set cover for $\S$ of cardinality at most $k$.
\end{lemma}

Due to space constraints, we redirect the reader to the Appendix for the proof.  In a nutshell, given a sequence of events from $G$ to a genome with CNP $\c$, the idea is to find, for each $u \in U$, one occurrence of $\e{u}$ in $G$ that we have control over.  More precisely, even though that occurrence of $\e{u}$ might spawn duplicates, all its copies (and copies of copies, and so on) will eventually get deleted.  The $\s{S_i}$ character preceding this $\e{u}$ character indicates that $S_i$ can be added to a set cover.  The crux of the proof is to show that this $\e{u}$ character exists for each $u \in U$, and that their corresponding $\s{S_i}$ form a set cover of size at most $k$.

We arrive to our main inapproximability result, which again follows from Lemma~\ref{lem:general-inapprox}.

\begin{theorem}
Assuming P $\neq$ NP, there is no polynomial-time constant factor approximation algorithm for~\mcng.
\end{theorem}

In the next section, we prove that the MCNG problem, parameterized by the
solution size, is W[1]-hard. This answers another open question
in \cite{qingge2018}. We refer readers for more details on FPT and
W[1]-hardness to the book by Downey and Fellows \cite{downey2012parameterized}.

\section{W[1]-hardness for MCNG}

Since SET-COVER is W[2]-hard, naturally we would like to use the ideas from the above reduction to prove the W[2]-hardness of~\mcng.  However, the fact that we use $t$-SET-COVER with constant $t$ in the proof of Lemma~\ref{lem:general-inapprox} is crucial, and $t$-SET-COVER is in FPT.  On the other hand, the property that is really needed in the instance of this proof, and in out~\mcng~ reduction, is that we can transform any set cover instance into an exact cover.  We capture this intuition in the following, and show that SET-COVER instances that have this property are W[1]-hard to solve. 

An instance of SET-COVER-with-EXACT-COVER, or SET-COVER-EC for short, is a pair $I = (\S, k)$ where $k$ is an integer and 
$\S$ is a collection of sets forming a universe $U$.
In this problem, we require that $\S$ satisfies the property that \emph{any} set cover for $\S$ of size at most $k$ is also an exact cover.
We are asked whether there exists a set cover for $\S$ of size at most $k$ (in which case this set cover is also an exact cover).
Therefore, SET-COVER-EC is a promise problem.

\begin{lemma}
The SET-COVER-EC problem is W[1]-hard for parameter $k$.
\end{lemma}

\begin{proof}
We show W[1]-hardness using the techniques introduced by Fellows et al. which is coined as MULTICOLORED-CLIQUE \cite{fellows2009}.
In the MULTICOLORED-CLIQUE problem, we are given a graph $G$, an integer $k$ and a coloring $c:V(G) \rightarrow [k]$ such that no two vertices of the same color share an edge.  We are asked whether $G$ contains a clique of $k$ vertices (noting that such a clique must have a vertex of each color).  This problem is W[1]-hard w.r.t. $k$.

Given an instance $(G, k, c)$ of MULTICOLORED-CLIQUE, we construct an instance $I = (\S, k')$ of SET-COVER-EC.  We put $k' = k + {k \choose 2}$.
For $i \in [k]$, let $V_i = \{v \in V(G) : c(v) = i\}$ and for each pair $i < j \in [k]$, 
let $E_{ij} = \{uv \in E(G) : u \in V_i, v \in V_j\}$.
The universe $U$ of the SET-COVER-EC instance has one element for each color $i$, one element for each pair $\{i,j\}$ of distinct colors, and two elements for each edge, one for each direction of the edge.  That is, 

$$U = [k] \cup {[k] \choose 2} \cup \{(u,v) \in V(G) \times V(G) : uv \in E(G)\}$$

Thus $|U| = k + {k \choose 2} + 2|E(G)|$.
For two colors $i < j \in [k]$, we will denote $U_{ij} = \{(u,v), (v,u) : u \in V_i, v \in V_j, uv \in E_{ij} \}$, i.e. we include in $U_{ij}$ both elements corresponding to each $uv \in E_{ij}$.
Now, for each color class $i \in [k]$ and each vertex $u \in V_i$, add to $\S$ the set

$$S_u = \{i\} \cup \{(u, v) : v \in N(u)\}$$

where $N(u)$ is the set of neighbors of $u$ in $G$.
Then for each $i < j \in [k]$, and for each edge $uv \in E_{ij}$, add to $\S$ the set

%\begin{align*}
%S_{uv} = \{\{i,j\}\} &\cup \{(x, y) : xy \in E_{ij}, x \in V_i, x \neq u\} \\
%&\cup \{(y, x) : xy \in E_{ij}, y \in V_j, y \neq v\}
%\end{align*}

$$
S_{uv} = \{\{i,j\}\} \cup \{(x,y) \in U_{ij} : x \notin \{u,v\}\}
$$

\begin{figure}

\centering
\includegraphics[scale=0.8]{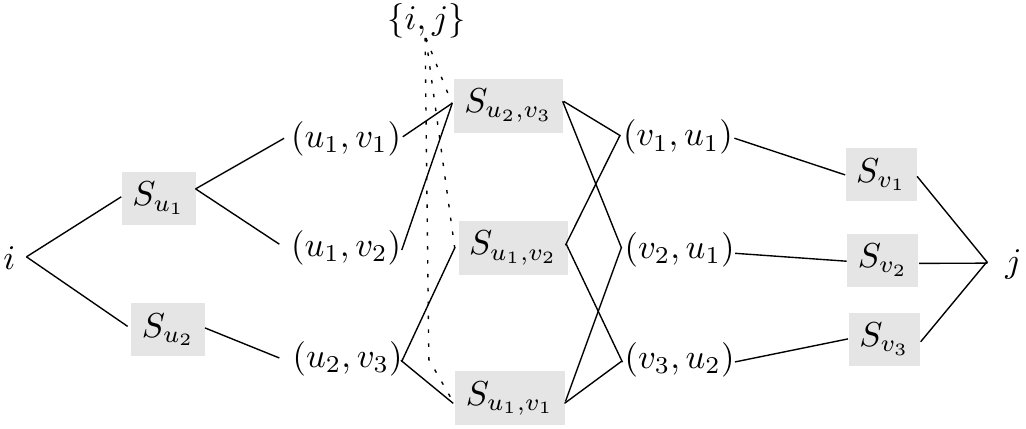}
\caption{A graphical example of the constructed sets for the $U_{ij}$ elements of a graph (not shown) with $E_{ij} = \{u_1v_1, u_1v_2, u_2v_3\}$, where the $u_l$'s are in $V_i$ and the $v_l$'s in $V_j$ (sets have a gray background, edges represent containment, the $\{i,j\}$ lines are dotted only for better visualization).}\label{fig:reduction_w1}

\end{figure}

The idea is that $S_{uv}$ can cover every element of $U_{ij}$, except those ordered pairs whose first element is $u$ or $v$.  Then if we do decide to include $S_{uv}$ in a set cover, it turns out that we will need to include $S_u$ and $S_v$ to cover these missing ordered pairs. See Figure~\ref{fig:reduction_w1} for an example.  For instance if we include $S_{u_2, v_3}$ in a cover, the uncovered $(u_2, v_3)$ and $(v_3, u_2)$ can be covered with $S_{u_2}$ and $S_{v_3}$.
We show that $G$ has a multicolored clique of size $k$ if and only if $\S$ admits a set cover of size $k'$.
Note that we have not shown yet that $(\S, k')$ is an instance of SET-COVER-EC, i.e. that any set cover of size at most $k'$ is also an exact cover.  This will be a later part of the proof.

First suppose that $G$ has a multi-colored clique $C = \{v_1, \ldots, v_k\}$, where $v_i \in V_i$ for each $i \in [k]$.  Consider the collection 
$$
\S^* = \{S_{v_1}, \ldots, S_{v_k}\} \cup \{S_{v_iv_j} : v_i, v_j \in C, 1 \leq i < j \leq k\},
$$ 
the cardinality of $\S^*$ is $k + {k \choose 2} = k'$.  Each element $i \in U \cap [k]$ is covered since we include a set $S_{v_i}$ for each color. 
Each element $\{i, j\} \in U \cap {[k] \choose 2}$ is covered since we include a set $S_{v_iv_j}$ for each color pair $i, j$ with $i < j$.
Consider an element $(x_i, y_j) \in U \cap (V(G) \times V(G))$, where $x_i \in V_i$ and $y_i \in V_j$.  Note that either $i < j$ or $j < i$ is possible, and that 
$v_iv_j \in E(G)$.
If $x_i \notin \{v_i, v_j\}$, then $S_{v_iv_j}$ covers $(x_i, y_j)$.  
If $x_i = v_i$, then $S_{v_i}$ covers $(x_i, v_j)$ and if $x_i = v_j$, then 
$S_{v_j}$ covers $(x_i, v_j)$.  Thus $\S^*$ is a set cover, and is of size at most $k'$.

For the converse direction, suppose that $\S^*$ is a set cover for $\S$ of size at most $k' = k + {k \choose 2}$.
Note that to cover the elements of $U \cap [k]$, $\S^*$ must have at least one set $S_u$ such that $u \in V_i$ for each color class $i \in [k]$.  Moreover, to cover the elements of $U \cap {[k] \choose 2}$, $\S^*$ must have at least one set $S_{uv}$ such that $u \in V_i, v \in V_j$ for each $i, j \in [k]$ pair.
We deduce that $\S^*$ has exactly $k + {k \choose 2}$ sets.  Hence for color $i \in [k]$, there is \emph{exactly} one set $S_u$ in $\S^*$ for which $u \in V_i$, and for each $\{i, j\}$ pair, there is \emph{exactly} one $S_{uv}$ set 
in $\S^*$ for which $u \in V_i, v \in V_j$.  

We claim that $C = \{u : S_u \in \S^*\}$ is a multi-colored clique.  We already know that $C$ contains one vertex of each color.  Now, suppose that some $u, v \in C$ do not share an edge, where $u \in V_i, v \in V_j$ and $i < j$.  Let $S_{xy}$ be the set of $\S^*$ that covers $\{i,j\}$, with $x \in V_i, y \in V_j$.  Since $uv$ is not an edge but $xy$ is, we know that $u \neq x$ or $v \neq y$ (or both).
Moreover, $S_{xy}$ does not cover the $(x, y)$ and $(y, x)$ elements of $U_{ij}$, and we know that at least one of these is not covered by $S_u$ nor $S_v$ (if $u \neq x$, then none covers $(x,y)$, if $v \neq y$, then none covers $(y,x)$).  But $(x, y) \in U_{ij}$, and $S_u, S_v$ and $S_{xy}$ are the only sets of $\S^*$ that have elements of $U_{ij}$, contradicting that $\S^*$ is a set cover.   This shows that $C$ is a multi-colored clique.

It remains to show that $\S^*$ is an exact cover.  
Observe that no two distinct $S_u$ and $S_v$ sets in $\S^*$ can intersect because $u$ and $v$ must be of a different color, and no two distinct $S_{uv}$ and $S_{xy}$ sets in $\S^*$ can intersect because $\{u, v\}$ and $\{x, y\}$ must be from two different color pairs. 
Suppose that $S_u, S_{xy} \in \S^*$ do intersect, and say that $x \in V_i, y \in V_j$ and $i < j$.  Then all elements in $S_u \cap S_{xy}$ are of the form $(u, v)$ for some $v$.  Choose any such $(u,v)$.  If $u$ is of color $i$, then $u \neq x$ since otherwise by construction $S_{xy}$ could not contain $(u,v)$.  But when $u \neq x$, no set of $\S^*$ covers the element $(x,y)$ (it is not $S_u$ nor $S_{xy}$, the only two possibilities).  If $u$ is of color $j$, then $u \neq y$ since again $S_{xy}$ could not contain $(u, v)$.  In this case, no set of $\S^*$ covers $(y,x)$.
We reach a contradiction and deduce that $\S^*$ is an exact cover.
%\qed
\end{proof}

It is now almost immediate that~\mcng~is W[1]-hard with respect to the natural parameter, namely the number of events to transform a genome $G$ into a genome with a given profile $\c$ (proof in Appendix). 

\begin{theorem}\label{thm:w1-hard}
The~\mcng~problem is W[1]-hard.
\end{theorem}

Now that we have finished presenting the negative results on MCNG.
An immediate question is whether we could obtain some positive result
on a related problem. In the next section, we present some positive
result for an interesting variation of MCNG.

\section{The Copy Number Profile Conforming Problem}

%\todo[inline]{Not a major thing, but if we want to be consistent with notation, all vectors should use the same notation as above, e.g. use $\vec{c_1},\vec{c_2}, \vec{v}$ instead of $C_1, C_2, V$.
%Or if we're lucky, we can just change the $\c$ latex command to output $C$ to change the above notation (but this'll require checking).}

We define the more general \emph{Copy Number Profile Conforming (CNPC)}
problem as follows: 
\begin{definition}
Given two CNP's $\c_1=\langle u_1,u_2,...,u_n\rangle$ and
$\v_2=\langle v_1,v_2,...,v_n\rangle$, 
%with $u_i,v_i\geq 0$ and $u_i,v_i\in \mathbb{N}$,
the CNPC problem asks to compute two strings $S_1$ and $S_2$ with $cnp(S_1)=\c_1$ and $cnp(S_2)=\c_2$
such that the distance between $S_1$ and $S_2$, $d(S_1,S_2)$, is minimized.
\end{definition}

Let $\sum_i u_i = m_1, \sum_i v_i = m_2$, we assume that $m_1$ and $m_2$ are bounded
by a polynomial of $n$. (This assumption is needed as the solution of our algorithm
could be of size $\max\{m_1,n_2\}$.) We simply say $\c_1,\c_2$ are polynomially bounded. 
Note that $d(S_1,S_2)$ is a very general distance measure, i.e., it could be
any genome rearrangement distance (like reversal, transposition, and tandem
duplication, etc, or their combinations, e.g. tandem duplication + deletion).
In this paper, we use the breakpoint distance (and the adjacency number),
which is defined as follows. (These definitions are adapted from Angibaud
et al. \cite{angibaud2009} and Jiang et al. \cite{jiang2012}, which generalize
the corresponding concepts on permutations \cite{watterson1982}.)

Given two sequences $A$=$a_{1}a_{2}\cdots a_{n}$ and $B$=$b_{1}b_{2}\cdots b_{m}$,
if $\{a_{i},a_{i+1}\}$ = $\{b_{j},b_{j+1}\}$
we say that $a_{i}a_{i+1}$ and $b_{j}b_{j+1}$ are matched to each other.
In a maximum matching of 2-substrings in $A$ and $B$, a matched pair is called
an {\em adjacency}, and an unmatched pair is called a {\em breakpoint} in $A$ and $B$ respectively. Then, the number of breakpoints in $A$ (resp. $B$)
is denoted as $d_b(A,B)$ (resp. $d_b(B,A)$), and the number of (common)
adjacencies between $A$ and $B$ is denoted as $a(A,B)$. For example, if
$A=acbdcb, B=abcdabcd$, then $a(A,B)=3$ and there are 2 and 4 breakpoints
in $A$ and $B$ respectively. 

Coming back to our problem, we define $d(S_1,S_2)=d_b(S_1,S_2)+d_b(S_2,S_1)$.
 From the definitions, we have 
$$d_b(S_1,S_2)+d_b(S_2,S_1)+2\cdot a(S_1,S_2)=(m_1-1)+(m_2-1),$$ or, 
$$d_b(S_1,S_2)+d_b(S_2,S_1) = m_1+m_2-2\cdot a(S_1,S_2)-2.$$
Hence, the problem is really to maximize $a(S_1,S_2).$

\begin{definition}
Given $n$-dimensional vectors $\u=\langle u_1,u_2,...,u_n\rangle$ and
$\w=\langle w_1,w_2,...,w_n\rangle$, with $u_i,w_i\geq 0$, and $u_i,w_i\in \mathbb{N}$,
we say $\w$ is a {\em sub-vector} of $\u$  if $w_i\leq u_i$ for $i=1,...,n$, also denote this relation as $\w\leq \u$.
\end{definition}
Henceforth, we simply call $\u,\w$ integer vectors (with the understanding that
no item in a vector is negative).
\begin{definition}
Given two $n$-dimensional integer vectors $\u=\langle u_1,u_2,...,u_n\rangle$
and $\v=\langle v_1,v_2,...,v_n\rangle$, with $u_i,v_i\geq 0$, and
$u_i,v_i\in \mathbb{N}$, we say $\w$ is a {\em common sub-vector} of $\u$ and
$\v$ if $\w$ is a sub-vector of $\u$ and $\w$ is also a sub-vector of $\v$
(i.e., $\w\leq \u$ and $\w\leq \v$). Finally, $\w$ is the {\em maximum common
sub-vector} of $\u$ and $\v$ if there is no common sub-vector $\w'\neq \w$ of
$\u$ and $\v$ which satisfies $\w\leq \w'\leq \u$ or $\w\leq \w'\leq \v$.
\end{definition}

An example is illustrated as follows. We have $\u=\langle 3,2,1,0,5\rangle$,
$\v=\langle 2,1,3,1,4\rangle$, $w'=\langle 2,1,0,0,3\rangle$ and
$\w=\langle 2,1,1,0,4\rangle$. Both $\w$ and $\w'$ are common sub-vectors for
$\u$ and $\v$, $\w'$ is not the maximum common sub-vector of $\u$ and $\v$
(since $\w'\leq \w$) while $\w$ is.

Given a CNP $\u=\langle u_1,u_2,...,u_n\rangle$ and alphabet $\Sigma=(x_1,x_2,...,x_n)$, for $i \in \{1,2\}$, we use $S(\u)$ to denote the multiset of letters (genes)
corresponding to $\u$; more precisely, $u_i$ denotes the number of $x_i$'s in
$S(\u)$. Similarly, given a multiset of letters $Z$, we use $s(Z)$ to denote
a string where all the letters in $Z$ appear exactly once (counting
multiplicity; i.e, $|Z|=|s(Z)|$). $s(Z)$ is similarly defined when $Z$ is a CNP.
We present Algorithm 1 as follows.
\begin{enumerate}
\item Compute the maximum common sub-vector $\v$ of $\c_1$ and $\c_2$.
\item Given the gene alphabet $\Sigma$, compute $S(\v)$, $S(\c_1)$ and $S(\c_2)$.
Let $X=S(\c_1)-S(\v)$ and $Y=S(\v_2)-S(\v)$.
\item If $S(\v)=\emptyset$, then return two arbitrary strings $s(\c_1)$ and $s(\c_2)$
as $S_1$ and $S_2$, exit; otherwise, continue.
\item Find $\{x,y\}$, $x,y\in\Sigma$ and $x\neq y$, such that $x\in S(\v)$ and $y\in S(\v)$,
and exactly one of $x,y$ is in $X$ (say $x\in X$), and the other is in $Y$
(say $y\in Y$). If such an $\{x,y\}$ cannot be found then return two strings
$S_1$ and $S_2$ by concatenating letters in $X$ and $Y$ arbitrarily at the
ends of $s(\v)$ respectively, exit; otherwise, continue.
\item Compute an arbitrary sequence $s(\v)$ with the constraint that
the first letter is $x$ and the last letter is $y$. Then obtain
$s_1=s(\v)\circ x$ and $s_2=y\circ s(\v)$ ($\circ$ is the concatenation operator).
\item Finally, insert all the elements in $X-\{x\}$ arbitrarily at the two ends 
of $s_1$ to obtain $S_1$, and insert all the elements in $Y-\{y\}$ arbitrarily at the two ends of $s_2$ to obtain $S_2$.
\item Return $S_1$ and $S_2$.
\end{enumerate}

Let $\Sigma=\{a,b,c,d,e\}$. Also let $\c_1=\langle 2,2,2,4,1\rangle$ and
$\c_2=\langle 4,4,1,1,1\rangle$. We walk through the algorithm using
this input as follows.

\begin{enumerate}
\item The maximum common sub-vector $\v$ of $\c_1$ and $\c_2$ is $\v=\langle 2,2,1,1,1\rangle$.
\item Compute $S(\v)=\{a,a,b,b,c,d,e\}$, $S(\c_1)=\{a,a,b,b,c,c,d,d,d,d,e\}$ and\\
$S(\c_2)=\{a,a,a,a,b,b,b,b,c,d,e\}$. Compute $X=\{c,d,d,d\}$ and $Y=\{a,a,b,b\}$.
\item Identify $d$ and $a$ such that $d\in S(\v)$ and $a\in S(\v)$, and $d\in X$ while
$a\in Y$.
\item Compute $s(\v)=dabbcea$, $s_1=dabbcea\cdot d$ and $s_2=a\cdot dabbcea$.
\item Insert elements in $X-\{d\}=\{c,d,d\}$ arbitrarily at the
right end of $s_1$ to obtain $S_1$, and insert all the elements in $Y-\{a\}=\{a,b,b\}$ 
at the right end of $s_2$ to obtain $S_2$.
\item Return $S_1=dabbcea\cdot d\cdot cdd$ and 
$S_2=a\cdot dabbcea\cdot abb$.
\end{enumerate}

\begin{theorem}
Let $\c_1,\c_2$ be polynomially bounded. The number of common adjacencies generated by Algorithm 1 
is optimal with a value either $n^*$ or $n^*-1$, where $n^*=\sum_{i=1}^{n}v_i$ with the 
maximum common sub-vector of $\c_1$ and $\c_2$ being $\v=\langle v_1,v_2,...,v_n\rangle$.
\end{theorem}

\begin{proof}
First, note that if $\v$ is a 0-vector (or $S(\v)=\emptyset$) then there will
not be any adjacency in $S_1$ and $S_2$. Henceforth we discuss $S(\v)\neq\emptyset$.

Notice that a common adjacency between $S_1$ and $S_2$ must come from two
letters which are both in $S(\v)$. That naturally gives us $n^*-1$ adjacencies,
where $n^*=|S(\v)|$, which can be done by using the letters in $S(\v)$ to form
two arbitrary strings $S_1$ and $S_2$ (for which $s(\v)$ is a common substring).
If $\{x,y\}$ can be found such that $x,y\in S(\v)$ and
$x\neq y$, and one of them is in $X$ (say $x\in X$), and the other is
in $Y$ (say $y\in Y$), then, obviously we could obtain $s_1=s(\v)\circ x$ and
$s_2=y\circ s(\v)$ which are substrings of $S_1$ and $S_2$ respectively. 
Clearly, there are $n^*=|S(\v)|$ adjacencies between $s_1$ and $s_2$ (and
also $S_1$ and $S_2$). 

To see that this is optimal, first suppose that no $\{x, y\}$ pair as above can be found.  
This can only occur when there are no two components $i<j$ in 
$\c_1=\langle c_{1,1},...,c_{1,i},...,c_{1,j}$,..., $c_{1,n}\rangle$,
$\c_2=\langle c_{2,1},...,c_{2,i},...,c_{2,j}$,..., $c_{2,n}\rangle$, and in the maximum common sub-vector $\v=\langle v_1,...,v_i$,..., $v_j,...,v_n\rangle$ of $\c_1$ and $\c_2$ which satisfy that
$\min\{c_{1,i},c_{2,i}\}=v_i\neq 0$ and $\max\{c_{1,i},c_{2,i}\}\neq v_i$, and
$\min\{c_{1,j},c_{2,j}\}=v_j\neq 0$ and $\max\{c_{1,j},c_{2,j}\}\neq v_j$. If this condition
holds, then all the components $i$ in $s(\c_1-\v)$ and $s(\c_2-\v)$, i.e., $c_{1,i}-v_i$ and $c_{2,i}-v_i$, have the property that at least one of the two is zero and $v_i=0$. Therefore, except for the letters corresponding to $\v$, no other adjacency can be formed. As any string with CNP $\v$ has $n^*$ characters, at most $n^*-1$ adjacencies can be formed.  
If an $\{x, y\}$ pair can be found, let $b \in \Sigma$, and let $v_b$ be the minimum copy-number of $b$ in $\c_1$ or $\c_2$, i.e., $v_b=\min\{c_{1,b},c_{2,b}\}$. Assume this minimum occurs in $\c_1$, w.l.o.g.  There can be at most $2v_b$ adjacencies involving $b$ in $\c_1$, and thus at most $2v_b$ adjacencies in common involving $v_b$.  Summing over every $b \in \Sigma$, the sum of common adjacencies, counted for each character individually, is at most $\sum_{b \in \Sigma}2v_b = 2n^*$.  Since each adjacency is counted twice in this sum, the number of common adjacencies is at most $n^*$.
%To see that this is optimal, first suppose that no $\{x, y\}$ pair as above can be found.  This can only occur when $\c_1 \leq \c_2$ or $\c_2 \leq \c_1$.  Suppose $\c_1 \leq \c_2$, w.l.o.g.  Then $\c_1 = \v$, meaning that any string with CNP $\c_1$ has $n^*$ characters, and thus that at most $n^* - 1$ adjacencies can be formed.  
%If an $\{x, y\}$ pair can be found, let $b \in \Sigma$, and let $v_b$ be the minimum copy-number of $b$ in $\c_1$ or $\c_2$.  Assume this minimum occurs in $\c_1$, w.l.o.g.  There can be at most $2v_b$ adjacencies involving $b$ in $\c_1$, and thus at most $2v_b$ adjacencies in common involving $v_b$.  Summing over every $b \in \Sigma$, the sum of common adjacencies, counted for each character individually, is at most $\sum_{b \in \Sigma}2v_b = 2n^*$.  Since each adjacency is counted twice in this sum, it follows that the number of common adjacencies is at most $n^*$.
%
%From the construction, we cannot obtain $n^* + 1$ adjacencies
%as $x\neq y$, trying to attach another copy of $y$ and $x$ to $s_1$ and $s_2$ would
%violate that exactly one of them is in $X$ and the other is in $Y$. Hence, $n^*$ is 
%an upper bound of the optimal solution.
\end{proof}

Note that if we only want the breakpoint distance between $S_1$ and $S_2$, then
the polynomial boundness condition of $\c_1$ and $\c_2$ can be withdrawn as we can
decide whether $\{x,y\}$ exists by searching directly in the CNPs (vectors).
%We comment that, in the case when $\{x,y\}$ can be found, it is possible to
%to construct a different solution. For instance, starting from $\v$,
%instead of using a common $s(\v)$ as the basis for constructing $S_1$ and
%$S_2$, we could use two different ones (with one breakpoint), say $s'(\v)$
%and $s"(v)$. In that case, we could attach elements in $X$ and $Y$ at the two
%ends of them to have two extra adjacencies. But the total number of adjacencies
%would still be $(n^*-1)-1+2=n^*$. We leave it as an exercise to find such
%an example for the interested readers. 

\section{Concluding Remarks}

In this paper, we answered two recent open questions regarding the computational
complexity of the Minimum Copy Number Generation problem. Our technique could
be used for other combinatorial optimization problems where the solution
involves Set Cover whose solution must also be an exact cover. We also present
a polynomial time algorithm for the Copy Number Profile Conforming (CNPC)
problem when the distance is the classical breakpoint distance. In some sense,
the breakpoint distance is static, and we leave open the question for solving
or approximating CNPC with any (dynamic) rearrangement distance (like
reversal, duplication+deletion, etc).

%\bibliographystyle{plain}
%\bibliography{main}

\newpage 

\section*{Appendix}

\subsection*{Proof of Lemma~\ref{lem:events-to-cover-dup}}

We need some new notation and intermediate results before proving the lemma.

Let $E = (e_1, \ldots, e_k)$ be a sequence of events transforming genome $G$ into another genome $G'$.
We would like to distinguish each position of $G$ in order to know which specific character of $G$ is at the origin of a character of $G'$.  

To that end, we augment each individual character of $G$ with a unique identifier, which is its position in $G$. 
That is, let $G = g_1g_2 \ldots g_n$, define a new alphabet $\hat{\Sigma} = (g_1^{1}, g_2^{2}, \ldots, g_n^{n})$ and define the genome $\hat{G} = g_1^{1}g_2^{2} \ldots g_n^{n}$.  
Here, two characters $g_i$ and $g_j$ may be identical, but $g_i^{i}$ and $g_j^{j}$ are two distinct characters.  We call $\hat{\Sigma}$ the \emph{augmented alphabet} and $\hat{G}$ the \emph{augmented genome} of $G$.  
%Conversely, if $\hat{G}'$ is a genome on the augmented alphabet, by removing the superscript we obtain a genome $G'$ on alphabet $\Sigma$.  We call $G''$ the \emph{shrunken genome} of $\hat{G''}$.
For instance if $G = aabcb$ and $\Sigma = (a, b, c)$, then $\hat{\Sigma} = (a^{1}, a^2, b^3, c^4, b^5)$ and $\hat{G} = a^{1} a^2 b^3 c^4 b^5$.

Since $G$ and $\hat{G}$ have the same length, we may apply the sequence $E$ on $\hat{G}$, resulting in a genome $\hat{G}' := \hat{G}\gevent{E}$ on alphabet $\hat{\Sigma}$.  Now $\hat{G}'$ may contain some characters of $\hat{\Sigma}$ multiple times owing to duplications, but if we remove the superscript identifier from the characters of $\hat{G}'$, we obtain $G'$.  The idea is that the identifiers on the characters of $\hat{G}'$ tell us precisely where each character of $\hat{G'}$ ``comes from'' in $\hat{G}$ (and thus $G$).

\begin{definition}
Let $G$ and $G'$ be genomes and let $E$ an event sequence such that $G' = G\gevent{E}$.  Let $\hat{G}$ be the augmented genome of $G$ and let $\hat{G}[i] = g^i$ be the character at position $i$.  

If there is at least one occurrence of $g^i$ in $\hat{G}\gevent{E}$, then position $i$ is called \emph{important} with respect to $E$. Otherwise, position $i$ is called \emph{unimportant} with respect to $E$.
\end{definition}

Roughly speaking, position $i$ is unimportant if it eventually gets deleted, and any character that was copied from position $i$ from a duplication also gets deleted, as well as a copy of this copy, and so on --- in other words, position $i$ has no ``descendant'' in $G'$ when applying $E$.

First, we prove some general properties that will be useful.  Recall that $G - s$ removes all occurrences of $s$ from $G$, and $\c - s$ puts $\c(s) = 0$.

\begin{proposition}\label{prop:replace-all-by-dont-care}
Let $G$ be a genome over alphabet $\Sigma$, let $\c$ be a CNP and let $s \in \Sigma$.
Then $\distgcnp(G - s, \c - s) \leq \distgcnp(G, \c)$. 
\end{proposition}

\begin{proof}
Let $E = (e_1, \ldots, e_k)$ be any optimal sequence of events changing $G$ into a genome $G'$ satisfying $\cnpg{G'} = \c$.  It is straightforward to see that $E$ can be adapted to a sequence $E' = (e_1', \ldots, e'_k)$ that transforms $G - s$ into $G' - s$ (it suffices to ensure that each $e'_i$ affects the same characters as $e_i$, with the exception of the $s$ characters - we omit the details).  Thus $\distgcnp(G - s, \c - s) \leq k = \distgcnp(G, \c)$.
%\qed
\end{proof}

\begin{proposition}\label{prop:replace-a-zero}
Let $G$ and $G^*$ be two genomes with $\distgg(G, G^*) = k$, and let $E$ be a sequence of $k$ events transforming $G$ into $G^*$.  Suppose that position $p$ is unimportant w.r.t. $E$, and let $G'$ be the genome obtained from $G$ by replacing $G[p]$ by any other character.  Then $\distgg(G', G^*) \leq \distgg(G, G^*)$.
\end{proposition}

\begin{proof}
Let $g = G[p]$, let $h \in \Sigma \setminus \{g\}$, and suppose that $G'$ is identical to $G$ except that $G'[p] = h$. Let $\hat{G}$ be the augmented $G$, with 
$\hat{G}[p] = g^p$.    
Moreover let $\hat{G'}$ be the augmented $G'$, so that $\hat{G'}[p] = h^p$.  Then the only difference between 
$\hat{G}\gevent{E}$ and $\hat{G'}\gevent{E}$ is that the occurrences of $g^p$ in $\hat{G}\gevent{E}$ are replaced by $h^p$ in $\hat{G'}\gevent{E}$. 
But since $p$ is unimportant, $\hat{G}\gevent{E}$ has no occurrence of $g^p$, meaning that $\hat{G}\gevent{E} = \hat{G'}\gevent{E}$.  Thus by removing the superscript identifiers from either $\hat{G}$ or $\hat{G}'$, we obtain $G^*$, implying that $G'\gevent{E} = G^*$.  It follows that $\distgg(G', G^*) \leq \distgg(G, G^*)$.
%\qed
\end{proof}

The next technical lemma states that if a genome alternates between positions to keep and positions to delete  $n$ times, then we need $n$ events to remove the unimportant ones.  

\begin{lemma}\label{lem:alternate-genomes}
Let $\Sigma = X \cup Y$ be an alphabet defined by two disjoint sets $X = \{x_1, \ldots, x_n\}$ and $Y$.
Let $G = Y_0 x_1 Y_1 x_2 Y_2 \ldots x_n Y_n $ be a genome on $\Sigma$, where for all $i \in [n]$, $Y_i$ is a non-empty string over alphabet $Y$ and $Y_0$ is a possibly empty string on alphabet $Y$.   Moreover let $\c$ be a CNP such that 
$\c(x_i) = 1$ for all $x_i \in X$ and $\c(y) = 0$ for all $y \in Y$.  
Then $\distgcnp(G, \c) \geq n$, with equality when $Y_0$ is empty.
\end{lemma}

\begin{proof}
It is easy to see that $d(G, \c) \leq n$ when $Y_0$ is empty, since one can delete every $Y_i$ substring one by one.  Thus we only need to show that $\distgcnp(G, c) \geq n$ whether $Y_0$ is empty or not, which we prove by induction over $n$.  When $n = 1$, this is obvious, as we need to at least delete the $Y_1$ segment.  For $n > 1$, take any sequence of events $E = (e_1, \ldots, e_k)$ that transforms $G$ into a genome $G^*$ with $\cnpg{G^*} = \c$, and let $G' := G\gevent{e_1}$.

If $e_1$ is a deletion, then $e_1$ cannot have deleted an $x_i$ character because there would be $0$ left whereas $\c(x_i) = 1$ (and we cannot create new copies of $x_i$ once its number of occurrences is $0$).  Thus $e_1$ could only delete some or all characters from a single $Y_i$ substring, $i \in \{0, \ldots, n\}$.  Let $Y'_i$ be the (possibly empty) substring that remains from $Y_i$ after the $e_1$ deletion. 
Note that if $i > 0$, then 
$Y'_i$ is preceded by the $x_i \in X$ character.  Define $x^* = x_i$ if $i > 0$, and $x^* = x_1$ if $i = 0$.   Consider the genome $G' - x^*$.  It has the form $Y_0x_1Y_1 \ldots Y_{i-1} Y'_i x_{i+1} Y_{i+1} \ldots x_nY_n$ if $i > 0$, and the form $Y_0'Y_1x_2Y_2 \ldots x_nY_n$ if $i = 0$.  In either case, the $x_i$'s and $Y$ substring alternate $n - 1$ times, and we may use induction on $G' - x^*$ to deduce that $\distgcnp(G' - x^*, \c - x^*) \geq n - 1$.
Moreover by Proposition~\ref{prop:replace-all-by-dont-care}, $\distgcnp(G' - x^*, \c - x^*) \leq \distgcnp(G', \c)$.  Thus $\distgcnp(G, \c) = \distgcnp(G', \c) + 1 \geq n - 1 + 1 = n$, as desired.

So suppose that $e_1$ is a duplication.  In $G'$, each occurrence of an $x_i \in X$ character is followed by a $Y$ character except at most possibly one.  In other words, after a duplication there can be at most one position $p$ such that both $G'[p]$ and $G'[p+1]$ are in $X$, which happens when we insert a duplicated substring between some $x_i$ and $Y_i$, or when we insert a duplicated substring between some $Y_i$ and $x_{i+1}$.  If it is the case that $G'[p]$ and $G'[p+1]$ are in $X$, then put $x_i = G[p]$. If no such position $p$ exists, then choose $x_i \in X$ arbitrarily.  

Now consider the string $G' - x_i$.  In $G' - x_i$, each $X$ character is preceded and followed by a $Y$ substring (possibly with the exception of $x_1$ which might have no preceding character).  By Proposition~\ref{prop:replace-all-by-dont-care} we know that $\distgcnp(G' - x_i, \c - x_i) \leq \distgcnp(G', \c)$.
We cannot use induction on $G' - x_i$ yet, as it is possible that some character in $X \setminus \{x_i\}$ has two occurrences in $G' - x_i$ owing to the $e_1$ duplication, and thus we do not have the form required for induction.  Suppose that $x_j \in X \setminus \{x_1\}$ does have two occurrences in $G' - x_i$.
Let $\tilde{E}$ be an optimal sequence of events turning $G' - x_i$ into a genome $G^*$ with copy-number profile $\c - x_i$.
Since $\c(x_j) = 1$, only one position of $G' - x_i$ containing character $x_j$ is important w.r.t. $\tilde{E}$.  So take the unimportant position of $G' - x_i$ having character $x_j$ and replace it by any character from $Y$, yielding another genome $\tilde{G}'$.  By Proposition~\ref{prop:replace-a-zero},  ${\distgcnp(\tilde{G}', \c - x_i) \leq \distgcnp(G' - x_i, \c - x_i)}$.  We repeat this for every $x_j \in X \setminus \{x_i\}$ having two occurrences, and call the resulting genome $\tilde{G}''$.  We note that ${\distgcnp(\tilde{G}'', \c - x_i) \leq \distgcnp(G' - x_i, \c - x_i)}$.  After relabeling if needed, $\tilde{G}''$ has the form $Y_0x_1Y_1x_2Y_2 \ldots x_{n-1}Y_{n-1}$.  Using induction, we know that 
\[n - 1 \leq \distgcnp(\tilde{G}'', \c - x_i) \leq \distgcnp(G' - x_i, \c - x_i) \leq \distgcnp(G', \c)\]

The lemma then follows from the fact that $\distgcnp(G', \c) = \distgcnp(G, \c) - 1$.
%\qed
\end{proof}

This allows us to prove the required assumption.

\noindent 
$\blacktriangleright$
\textbf{Lemma~\ref{lem:events-to-cover-dup}}.  \emph{
Given a sequence of $k$ events (deletions and duplications) transforming $G$ into a genome with CNP $\c$, one can obtain a set cover for $\S$ of cardinality at most $k$.
}

\begin{proof}
Suppose that the events $E = (e_1, \ldots, e_k)$  transform $G$ into a genome 
$G^*$ with CNP $\c$.  
We construct a set cover for $\S$ of cardinality $k$.
For a position $p$ with $G[p] = \e{u} \in \Sigma_U$, define $pred(p)$ as the first $\Sigma_{\S}$ character to the left of position $p$.
To be precise, if $p'$ is the largest integer satisfying $G[p'] \in \Sigma_{\S}$ and $p' < p$, then $pred(p) = G[p']$.  Note that since $G[1] = \s{S_1}$, $pred(p)$ is well-defined.
Notice that by construction, if $G[p] = \e{u}$ and $\s{S} = pred(p)$, then $u \in S$.
The set of $pred(p)$ of unimportant positions $p$ will correspond to our set cover, which we now prove by separate claims.

\begin{nclaim} 
For each $u \in U$, there is at least one position $p$ of $G$ such that ${G[p] = \e{u}}$ and such that $p$ is unimportant w.r.t. $E$.
\end{nclaim}

\begin{proof}
If we assume this is not the case, then each of the $f(u)$ positions $p$ of $G$ having ${G[p] = \e{u}}$ has a descendant in $G^*$, implying that $G^*$ has at least $f(u)$ copies of $\e{u}$ and thereby contradicting that $G^*$ complies with $\c(\e{u}) = f(u) - 1$.
%\qed
\end{proof}

%Note that there may be more than one unimportant position for $h_u$ due to possible duplications.

Given that the claim holds, let $P = \{p_1, \ldots, p_m\}$ be any set of positions of $G$ such that for each $i \in [m]$, $G[p_i] = \e{u_i}$ and $p_i$ is unimportant w.r.t. $E$ (choosing arbitrarily if there are multiple choices for $p_i$).
Define $\Sigma_P = \{pred(p_i) : p_i \in P\}$ and $\S^* = \{S_i \in \S : \s{S_i} \in \Sigma_P\}$.

\begin{nclaim} 
$\S^*$ is a set cover.
\end{nclaim}

\begin{proof}
For each $u_i \in U$, there is an unimportant position $p_i \in P$ such that $G[p_i] = \e{u_i}$.
Moreover, $pred(p_i)$ is some character $\s{S}$ such that $\s{S} \in \Sigma_P$ and such that $u_i \in S$.  Since $S \in \S^*$, it follows that each $u_i$ is covered.
%If we assume that the claim is false, then there is some $u_i \in U$ such that no set of $\S^*$ contains $u_i$.  This means that for every set $S \in \S$ that contains $u_i$, the unique position $p$ satisfying $G[p] = \e{u_i}$ and $pred(p) = \s{S}$ is important, i.e. it has a descendant in $G^*$.  Then there are at least $f(u_i)$ copies of $\e{u_i}$ in $G^*$, a contradiction.
%\qed
\end{proof}

It remains to show that $\S^*$ has at most $k$ sets. Denote ${P' = P \cup \{p : G[p] \in \Sigma_P\}}$. 
Let $\tilde{G}$ be the subsequence of $G$ obtained by keeping only the positions in $P'$ (i.e. if we denote $P' = \{p'_1, \ldots, p'_l\}$ with $p'_1 < p'_2 < \ldots < p'_l$, then $\tilde{G} = G[p'_1]G[p'_2] \ldots G[p'_l]$).
Furthermore, define the CNP $\tilde{\c}$ such that $\tilde{\c}(\s{S_i}) = 1$ for all $\s{S_i} \in \Sigma_P$, $\tilde{\c}(\s{S_i}) = 0$ for all $\s{S_i} \in \Sigma_{\S} \setminus \Sigma_P$, and 
$\tilde{\c}(\e{u}) = 0$ for all $\e{u} \in \Sigma_U$.  Note that $\tilde{G}$ has the form $\s{S_1}D_1\s{S_2}D_2 \ldots \s{S_r}D_r$ for some $r$, where the $D_i$'s are substrings over alphabet $\Sigma_U$.  This is the same form as in Lemma~\ref{lem:alternate-genomes}. 

\begin{nclaim}\label{claim:tilde-guys}
$\distgcnp(\tilde{G}, \tilde{\c}) \leq k$.
\end{nclaim}

\begin{proof}
Let $G'$ be the genome obtained by replacing every position $p$ of $G$ by some dummy character $\lambda$, except for the positions of $P'$ (thus if we remove all the $\lambda$ occurrences we obtain $\tilde{G})$.  Since $G$ and $G'$ have the same length, we can apply the $E$ events on $G'$.
Let $G'' := G'\gevent{E}$, and let $l$ be the number of occurrences of $\lambda$ in $G''$.  Recall that $P'$ contains only positions $p$ such that $G[p] \in \Sigma_{P}$, or such that $p$ is unimportant w.r.t $E$ and $G[p] \in \Sigma_U$.  It follows that if a position $q$ is important w.r.t. $E$, then $G'[q] \in \Sigma_{P} \cup \{\lambda\}$.  Moreover, for any $\s{S} \in \Sigma_P$, $G''$ has as many occurrences of $\s{S}$ as in $G\gevent{E}$.  In other words, $G''$ has one occurrence of each $\s{S} \in \Sigma_P$ and the rest is filled with $\lambda$.  

Let $\c'$ be the CNP satisfying 
$\c'(\lambda) = l$, $\c'(\s{S_i}) = \tilde{\c}(\s{S_i}) = 1$ for every $\s{S_i} \in \Sigma_P$, and $\c'(x) = 0$ for any other character $x$.
Then clearly, $\c' = \cnpg{G''}$, which implies $\distgcnp(G', \c') \leq k$ since $E$ transforms $G'$ into $G''$.
Moreover by Proposition~\ref{prop:replace-all-by-dont-care}, $\distgcnp(G' - \lambda, \c' - \lambda) \leq \distgcnp(G', \c') \leq k$.  The claim follows from the observation that $\tilde{G} = G' - \lambda$ and $\tilde{\c} = \c' - \lambda$.
%\qed
\end{proof}

Observe that $\tilde{G}$ and $\tilde{\c}$ have the required form for Lemma~\ref{lem:alternate-genomes} (with $|\Sigma_P|$ important positions), 
and so $\distgcnp(\tilde{G}, \tilde{\c}) \geq |\Sigma_P|$.  It follows from Claim~\ref{claim:tilde-guys} that $k \geq \distgcnp(\tilde{G}, \tilde{\c}) \geq |\Sigma_P| = |\S^*|$.  We have thus constructed a set cover $\S^*$ for $\S$ of cardinality at most $k$, which completes the proof.
%\qed
\end{proof}

\subsection*{Proof of Theorem~\ref{thm:w1-hard}}

%\begin{proof}
Let $(\S, k)$ be a SET-COVER-EC instance.  Construct a genome $G$ and copy-number vector $\c$ from $\S$ as described in Section~\ref{subsec:construction}. 
We show that $\S$ admits a set cover of size $k$ if and only if $\distgcnp(G, \c) \leq k$.  
If $\S$ admits a set cover $\S^*$ of size $k$, then $\S^*$ is also an exact cover and by Lemma~\ref{lem:exact-to-dels}, one can apply $k$ deletions on $G$ to obtain the desired vector $\c$.
Conversely, Lemma~\ref{lem:events-to-cover} a sequence of $k$ events for $(G, \c)$ show that $\S$ has a set cover of size at most $k$.
%\qed
%\end{proof}

\end{document}